\documentclass[letterpaper, 10 pt, conference]{ieeeconf}
\IEEEoverridecommandlockouts                             
\usepackage{amsfonts}
\usepackage{amsmath} % assumes amsmath package installed
\usepackage{amssymb}  % assumes amsmath package installed

\usepackage{float}

\usepackage{booktabs}
\usepackage{tabularx}
\usepackage{multirow}
\usepackage{hhline}
\usepackage{bm}
\usepackage{makecell}
\newcolumntype{C}[1]{>{\centering\let\newline\\\arraybackslash\hspace{0pt}}m{#1}}

\usepackage{amsthm}

\theoremstyle{definition}
\newtheorem{proposition}{Proposition}[section]

\usepackage{hyperref}
\hypersetup{
    colorlinks=true,
    linkcolor=blue,
    filecolor=magenta,      
    urlcolor=cyan,
}
\usepackage{tikz}
\usetikzlibrary{calc,patterns,decorations.pathmorphing,decorations.markings}
\usetikzlibrary{shapes,arrows}
\usepackage{verbatim}

\tikzstyle{block} = [draw, fill=blue!20, rectangle, 
    minimum height=3em, minimum width=6em]
\tikzstyle{sum} = [draw, fill=blue!20, circle, node distance=1cm]
\tikzstyle{input} = [coordinate]
\tikzstyle{output} = [coordinate]
\tikzstyle{pinstyle} = [pin edge={to-,thin,black}]

\usetikzlibrary{positioning}
\usetikzlibrary{math}

\usepackage{tikz}
\usetikzlibrary{shapes,arrows,calc,positioning}
\tikzstyle{bigblock} = [draw, fill=blue!20, rectangle, 
    minimum height=6em, minimum width=8em]
\tikzstyle{medblock} = [draw, fill=blue!20, rectangle, 
    minimum height=4em, minimum width=4em]    
\tikzstyle{mux} = [draw, fill=black!20, rectangle, 
    minimum height=5em, minimum width=0.1em]    
\tikzstyle{smallblock} = [draw, fill=blue!20, rectangle, 
    minimum height=3em, minimum width=4em]
\tikzstyle{sum} = [draw, fill=blue!20, circle, node distance=1cm]
\tikzstyle{signal} = [coordinate]
\tikzstyle{pinstyle} = [pin edge={to-,thin,black}]
\tikzstyle{block} = [draw, fill=blue!20, rectangle, 
    minimum height=3em, minimum width=6em]
\tikzstyle{blockS} = [draw, fill=blue!20, rectangle, 
    minimum height=3em, minimum width=4em]    
\tikzstyle{input} = [coordinate]
\tikzstyle{output} = [coordinate]
\usetikzlibrary{matrix}

\newcommand{\bc}{\begin{center}}
\newcommand{\ec}{\end{center}}
\newcommand{\benum}{\begin{enumerate}}
\newcommand{\eenum}{\end{enumerate}}
\newcommand{\nn}{\nonumber}
\newcommand{\matl}{\left[ \begin{array}}
\newcommand{\matr}{\end{array} \right]}
\newcommand{\matls}{\left[ \begin{smallmatrix}}
\newcommand{\matrs}{\end{smallmatrix} \right]}
\newcommand{\isdef}{\stackrel{\triangle}{=}}

\newcommand{\rmA}{{\rm A}}

\newcommand{\rmD}{{\rm D}}

\newcommand{\rmT}{{\rm T}}

\newcommand{\rmc}{{\rm c}}
\newcommand{\rmd}{{\rm d}}

\newcommand{\rmf}{{\rm f}}

\newcommand{\rmi}{{\rm i}}

\newcommand{\rmp}{{\rm p}}

\newcommand{\BBR}{{\mathbb R}}

\newcommand{\shiftq}{{\textbf{\textrm{q}}}}

\newcommand{\twolineentry}[2]{
\begin{tabular}{ c }
    #1 \\
    #2
\end{tabular}
}

\newcommand{\threelineentry}[3]{
\begin{tabular}{ c }
    #1 \\
    #2 \\
    #3
\end{tabular}
}

\title{An A Quadcopter Autopilot Based on an Adaptive Digital PID Controller}

\title{Adaptive Digital PID Control of a Quadcopter}

\title{Retrospective-Cost-Based Adaptive Digital PID Control of a Quadcopter}

\title{A Retrospective-Cost-Based Adaptive Digital PID   Quadcopter Autopilot}

\title{Adaptive Digital PID Control of a Quadcopter with Unknown Dynamics}

\title{One-Shot Learning for a Quadcopter Autopilot}

\title{An adaptive digital autopilot for Multicopters}

\title{\LARGE \bf Experimental Implementation of an Adaptive Digital Autopilot }

\title{\LARGE \bf An Adaptive PID Autotuner for Multicopters \\ with Experimental Results }

\author{
    John Spencer,
    Joonghyun Lee,
    Juan Augusto Paredes, 
    Ankit Goel, 
    Dennis Bernstein%
\thanks{This research was supported in part by the Office of Naval Research under grant N00014-19-1-2273.}
\thanks{John Spencer, Joonghyun Lee, Juan Augusto Paredes, and Dennis Bernstein are with the Department of Aerospace Engineering, University of Michigan, Ann Arbor, MI 48109.
{\tt\small spjohn, joonghle, jparedes,}
{\tt \small dsbaero@umich.edu}
}
\thanks{Ankit Goel is with the Department of Mechanical Engineering, University of Maryland, Baltimore County, Md 21250.
{\tt \small ankgoel@umbc.edu}
}
}

\date{September 2021}
\begin{document}

\maketitle

\begin{abstract}
    This paper develops an adaptive PID autotuner for multicopters, and presents simulation and experimental results. 
    The autotuner consists of adaptive digital control laws based on retrospective cost adaptive control implemented in the PX4 flight stack.
    A learning trajectory is used to optimize the autopilot during a single flight. 
    % and a trajectory that excites all control modes. 
    % 
    The autotuned autopilot is then compared with the default  PX4 autopilot by flying a test trajectory constructed using the second-order Hilbert curve.
    In order to investigate the sensitivity of the autotuner to the quadcopter dynamics, the mass of the quadcopter is varied, and the performance of the autotuned and default autopilot is compared. 
    % , which degrades the performance of the default autopilot with the default gains.
    % 
    It is observed that the autotuned autopilot outperforms the default autopilot.
    
    % gains provide better performance in comparison to the default gains. 
\end{abstract}

\section{Introduction}

Unmanned aerial vehicles, especially multicopters, have been used in a myriad of applications over the last decade, 
such as environment mapping, asset monitoring, risk assessment, sports broadcasting, wind-turbine inspection, and their applications continue to grow
\cite{chang2016development,anweiler2017multicopter,andaluz2015nonlinear, schafer2016multicopter,stokkeland2015autonomous,juan2017}.
In its most common form, a multicopter has four propellers. By controlling the spin rates of the four propellers, a force along a body-fixed axis and moments about three linearly independent body-fixed axes can be independently applied to affect desired translational as well as rotational motions.   
However, due to the nonlinear and unstable nature of the quadcopter dynamics, precise control of a quadcopter is a well-recognized challenging problem.

Nonlinear techniques such as feedback-linearization \cite{lee2009feedback} and back-stepping \cite{farrell2005backstepping} have been applied to deal with the nonlinearities and to construct stabilizing controllers. 
However, these techniques require accurate plant models at all operating conditions \cite{castillo2004real}. 
Adaptive techniques have also been investigated to reduce the need for an accurate model \cite{zuo2014adaptive,dydek2013MRAC}, however, they also require a sufficiently accurate plant model to construct stabilizing controllers.
Iterative learning control is used learn the pitch and roll controller in the closed-loop system  in \cite{abdolahi2018black}.
Reinforcement learning control is used to learn low-level controllers in case of multiple actuator failure in \cite{dooraki2020reinforcement}.
L1 adaptive control is used to improve the stability margins of a stable control loop in a quadcopter flight control system in  \cite{thu2016modeling}.
Fuzzy control was used in the position control in  \cite{tran2020fuzzy}.
Bidirectional brain emotional learning was used to improve trajectory tracking and handle payload uncertainties in \cite{kumar2021realtime}. 
Model reference adaptive control was used in the attitude controller in \cite{niit2017integration} to counteract model uncertainties. 
Adaptive twisting sliding mode control was used in the attitude controller in \cite{hoang2017adaptive} to resolve the chattering issues in standard sliding mode control. 
Robust fixed point transformation based adaptive control was used in  \cite{czako2017novel} to improve quadcopter stability in the presence of parameter uncertainties and external disturbances. 
Adaptive particle swarm optimization was used in a PD controller (APSO-PD) in \cite{yazid2018optimal} to improve the rise time, settling time, overshoot, and peak time of a standard PSO controller. 
Robust adaptive control based on backstepping is used in \cite{bhatia2019projection} \cite{kourani2018coping} to provide robust quadcopter altitude and attitude tracking  after payload change. 
Immersion and invariance based adaptive backstepping control is used in \cite{navabi2017immersion} to remedy quadcopter attitude instabilities due to disturbance torques or parameter uncertainties. 
Adaptive fault tolerant control based on the adaptive minimum projection method is used to provide quadcopter stability in the event of actuator failure in \cite{tabata2018adaptive}. 
Balanced control is used in \cite{hu2020research} to improve quadcopter performance by switching between fuzzy adaptive PID and optimized PID midflight. 
L1 adaptive control is used in \cite{zhang2020pitch} to stabilize a fixed-wing pitch controller, and in \cite{li2018lateral} to counteract fixed-wing split drag rudder damage. 
Adaptive control is also used in \cite{yu2018adaptive} to resolve instabilities due to trailing vortices in fixed-wing formation flight. 
Online adaptive model parameter estimation is used in the velocity controller in \cite{gui2018adaptive} to mitigate error due to model parameter uncertainty.
However, most of these control techniques focus on optimizing a part of the control system, 
%
%assuming the availability of a sufficiently good rest of the control system. 
%
assuming that the rest of the control system is sufficiently good. 

Typical quadcopter autopilots are, however, based on cascaded controllers, which consist of an inner loop to stabilize the dynamics, and an outer loop to track position commands. 
Traditionally, all controllers in the autopilot are constructed using manually tuned PID control laws.
In fact, widely used open-source autopilots such as PX4 and ArduPilot contain finely-tuned PID control laws for many commercially available multicopter configurations \cite{meier2015px4,ardupilot}. 
These autopilots cannot guarantee stability and thus have a fixed operational envelope, which is usually unknown. 
Moreover, autopilots tuned for a specific geometry and inertia properties do not perform well in the case where these properties vary with time, such as, in the case of unknown suspended payload, hardware alteration, and dynamic environmental changes. 
Such variations invariably degrade the performance of the autopilot.

% Earlier control systems were usually based on linearizing the dynamics of a quadcopter around the hover condition and utilizing linear quadratic theory to implement a stabilizing controller and integrators in the outer loops to track position commands. 
% % 
% However, the nonlinear nature of the dynamics makes guaraanting stability an impossible task.

With these motivations in mind, this paper develops an adaptive autotuner for multicopters with unknown dynamics. 
In particular, all of the fixed-gain control laws of an autopilot are replaced by adaptive control laws that optimize the gain by flying a single learning trajectory. 
Specifically, the adaptive controllers are updated by the retrospective cost adaptive control (RCAC) algorithm \cite{rezaPID}.
The learning trajectory is designed such that all possible motions of the quadcopter are excited and thus nonzero control outputs are required from all control laws, which ensures that all of the controller gains are adaptively updated.

The contribution of the work presented in this paper is 
the development of an adaptive autotuner that does not require any prior knowledge of the system dynamics, and instead uses a simple learning trajectory to adapt the autopilot gains; and its numerical and experimental demonstration.  
% 
% The performance of a finely-tuned autopilot is investigated by scaling the mass of the quadcopter in the simulation environment, and it is compared with the autotuned autopilot. 
% 
The performance of the autotuner is demonstrated by flying a test trajectory and comparing its performance with a finely-tuned autopilot. 
The improvements due to the autotuning process are demonstrated through simulations and flight tests.

The paper is organized as follows.
Section \ref{sec:PX4_autopilot} briefly presents the control system architecture implemented in the PX4 autopilot.
%
% Section \ref{sec:PID_Algo} presents the retrospective cost based adaptive control algorithm.
%
Section \ref{sec:autotuner} explains the setup of the adaptive autotuner.
Section \ref{sec:PID_Algo} describes the RCAC algorithm used to tune the control parameters.
% 
%Section \ref{sec:flight_tests} presents simulation and flight test results of the stock PX4 and the adaptive PX4 autopilot.
%
Section \ref{sec:flight_tests} presents simulation and flight test results showing the adaptation of the autotuner and comparing the performance of the PX4 autopilot using the stock and autotuned control parameters.
% 
% control architectures of stock PX4 autopilot and the adaptive PX4 autopilot are presented.
% In section \ref{sec:Numerical}, simulation results are presented to compare the performance of the adaptive PX4 autopilot with the stock PX4 autopilot.
%
Finally, section \ref{sec:conclusions} concludes the paper with a summary and future research directions.

\section{ Quadcopter Autopilot}
\label{sec:PX4_autopilot}
This section reviews the quadcopter autopilot considered in this work. 
This autopilot is based on the control architecture implemented in the PX4 autopilot.
% The PX4 autopilot architecture and 
The notation used in this paper is described in more detail in \cite{goel_adaptive_pid_2021}.

The control system consists of a mission planner and two nested loops as shown in Figure \ref{fig:PX4_autopilot_nested_loop}.
%
%The mission planner generates the position, velocity, azimuth, and azimuth rate setpoints from the user-defined waypoints.
%
The mission planner generates the position, velocity, and azimuth setpoints from the user-defined waypoints using the guidance law described below.
The outer loop consists of the \textit{position controller}, whose inputs are the position and velocity errors, defined as the difference between the setpoints and the measurements. 
% 
% as well as the estimated position and velocity.
%
The output of the position controller is the thrust vector setpoint.
Note that the thrust vector is expressed in terms of the Earth-fixed frame. 
The inner loop consists of the \textit{attitude controller}, whose inputs are the thrust vector setpoint, the azimuth and azimuth-rate setpoints, as well as the attitude and the angular rate measured in the body-fixed frame. 
The output of the attitude controller is the moment vector setpoint in the body-fixed frame. 
The magnitude of the thrust vector setpoint and the moment vector setpoint uniquely determine the required rotation rates of the four propellers. 

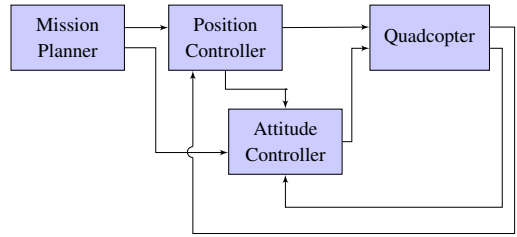
\begin{figure}[h]
    \centering
    \resizebox{0.8\columnwidth}{!}{
    \begin{tikzpicture}[auto, node distance=2cm,>=latex',text centered]
    
        \node [smallblock, minimum height=3em, text width=1.6cm] (Mission) {\small Mission Planner};
        \node [smallblock, minimum height=3em, right = 2em of Mission, text width=1.6cm] (Pos_Cont) {\small Position Controller};
        \node [smallblock, minimum height=3em, below right = 1.75em and -2.5 em of Pos_Cont,text width=1.6cm] (Att_Cont) {\small Attitude Controller};
        \node [smallblock, minimum height=3em, minimum width = 5.5em,  right = 4em of Pos_Cont] (Quadcopter) {\small Quadcopter};
        
        \draw [->] (Mission.10) -- 
        % node[above, xshift = -0.05 em, yshift = 0.1em]{\scriptsize$\begin{array}{c}\vect r_{c/w}\resolvedin{E}{\rm sp}, \\ \framedotE{\vect r}_{c/w}\resolvedin{E}{\rm sp,ff} \end{array}$} 
        (Pos_Cont.170);
        \draw[->] (Mission.350) -- +(.5,0) |- 
        % node[below, xshift = 1 em, yshift = 0.05em]{\scriptsize$\Psi_{\rm sp}, \dot \Psi_{\rm sp,ff}$}
        ([yshift = -0.5em]Att_Cont.180);

        \draw [->] (Pos_Cont.-90) |- +(1,-.3 ) -|
        % -| ([xshift = 0.75em, yshift = -2em]Pos_Cont.-10) -| 
        % node  [xshift = 3.5em, yshift = 1em] {\scriptsize$\vect f_\rmc\resolvedin{E}{\rm sp}$} ([xshift = -0.75em, yshift = 0.5em]Att_Cont.west) -- 
        (Att_Cont.90);

        \draw [->] (Pos_Cont.10) -- 
        % node [above, xshift = 0.05em, yshift = 0.1 em] {\scriptsize$\left|\left|\vect f_\rmc\resolvedin{E}{\rm sp}\right|\right|_2$} 
        (Quadcopter.170);
        \draw [->] (Att_Cont.0) -- +(0.15,0) |- 
        % node [below, xshift = 1.7 em, yshift = -2.75 em]{\scriptsize$\vect M_{\SQ/\rmc}\resolvedin{Q}{\rm sp}$}
        (Quadcopter.190);
        \draw [-] (Quadcopter.10) -- +(.4,0) |- 
        % node[below,xshift = -6em, yshift = 0.2em]{\scriptsize $\vect r_{c/w} \resolvedin{E} {\rm meas}, \framedotE{\vect r}_{c/w}\resolvedin{E} {\rm meas}$} 
        ([xshift = -1.5em, yshift = -7.5 em]Pos_Cont.south) -- ([xshift = -1.5em, yshift = -4 em]Pos_Cont.south);
        \draw ([xshift = -1.5em, yshift = -4 em]Pos_Cont.south) arc (270:90:0.25em);
        \draw [->] ([xshift = -1.5em, yshift = -3.5 em]Pos_Cont.south) -- ([xshift = -1.5em]Pos_Cont.south);
        \draw [->] (Quadcopter.350) -- +(0.2,0) |- 
        % node[below,xshift = -6em]{\scriptsize $q_{\rm Q/E}^ {\rm meas}, \vect \omega_{\rm Q/E} \resolvedin{Q}{\rm meas}$} 
        ([yshift = -1.5 em]Att_Cont.south) -- (Att_Cont.south);
        
    \end{tikzpicture}
    }
    %\vspace{-2em}
    \caption{\footnotesize Control-system architecture. }
    \label{fig:PX4_autopilot_nested_loop}
\end{figure}

The mission planner uses a guidance law described in Appendix \ref{app:guid_law} to generate the position and the azimuth setpoints. 
Using a user-specified maximum velocity $v_{\rm max}$ and maximum acceleration $a_{\rm max}$, the guidance law generates a trajectory that consists of a constant acceleration phase, a cruise phase, and a constant deceleration phase.
A similar guidance law is used to generate azimuth setpoints, given a user-specified maximum angular velocity $\dot{\psi}_{\rm max}$ and maximum angular acceleration $\ddot{\psi}_{\rm max}.$

The position controller consists of two cascaded linear controllers as shown in Figure \ref{fig:PX4_autopilot_outer_loop}.
The first controller $G_r$ consists of three proportional controllers.
The second controller $G_v$ consists of three decoupled PID controllers and a velocity setpoint feedforward controller, and yields the thrust vector setpoint.
% There are thus 12 ( 3 $\times$ Proportional + 3 $\times$ PID) scalar gains in the outer loop. 

\begin{figure}[h]
    \centering
    \resizebox{0.95\columnwidth}{!}{
    \begin{tikzpicture}[auto, node distance=2cm,>=latex']
        \node (Ref_traj) {};
    	\node [sum, right = 2 em of Ref_traj] (sum1) {};
    	\node[draw = white] at (sum1.center) {$+$};
    	\node [smallblock, minimum width = 2em, minimum height = 1.75 em, right = 2 em of sum1] (Cont1) {\small$G_r$};
    	% \node [sum, right of=Cont1,node distance=1.5cm] (sum2) {};
    	\node [sum, right = 2 em of Cont1] (sum3) {};
    	\node[draw = white] at (sum3.center) {$+$};
    	\node [smallblock, minimum width = 2em, minimum height = 1.75 em, right =2 em of sum3] (Cont2) {\small$G_v$};
    	\node [right = 2 em of Cont2] (output) {};

    	\draw [->] (Ref_traj) node [above, xshift=0.25em]
    	{\scriptsize  \twolineentry{Position}{setpoint}	}-- (sum1);
    	
    	\draw [->] ([yshift = -2em]sum1.south) -- node[xshift = .7em, yshift = -0.75em]{\scriptsize \twolineentry{Position}{measurement}}(sum1.south);
    	
    	\draw [->] (sum1) -- node [xshift=-0.7em, yshift = -1.5em]{$-$} (Cont1);
    	\draw [->] ([yshift = 2.5em]sum3.north) -- node [xshift = -4.0em, yshift = 1.em]{\scriptsize \twolineentry{Velocity}{setpoint}}(sum3.north);
    	\draw [->] (Cont1) -- (sum3);
    	\draw [->] (sum3) -- node [xshift=-0.7em, yshift = -1.42em]{$-$} (Cont2);
    	\draw [->] (Cont2) -- node [above,xshift = 0.5em, yshift = -0.1em] {\scriptsize \threelineentry{Thrust}{vector}{setpoint}} (output);
    	
    	\draw [->] ([yshift = -2em]sum3.south) -- node[xshift = 0.7em, yshift = -0.75em]{\scriptsize \twolineentry{Velocity}{measurement}}(sum3.south);

    \end{tikzpicture}
    }
    \vspace{-1em}\caption{\footnotesize Position controller architecture. }
    \label{fig:PX4_autopilot_outer_loop}
\end{figure}
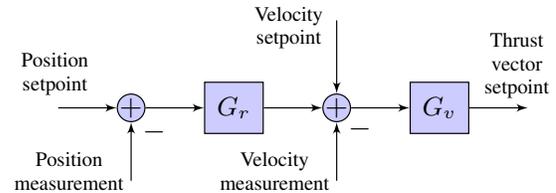

The force vector setpoint along with the azimuth setpoint are used to calculate the attitude setpoint, which is represented as a quaternion.  
The attitude controller consists of two cascaded controllers $G_q$ and $G_\omega$ as shown in Figure \ref{fig:PX4_autopilot_inner_loop}. 
% The static map f2q converts the force setpoint into the quaternion setpoint as described below.
% 
The first controller $G_q$ is an almost globally stabilizing controller  \cite{Chaturvedi2011} that consists of three proportional gains.
The second controller $G_\omega$ consists of three PID controllers and an angular rate setpoint feedforward controller, and yields the moment vector setpoint.

\begin{figure}[h]
    \centering
    \resizebox{0.95\columnwidth}{!}{
    \begin{tikzpicture}[auto, node distance=2cm,>=latex']
    %     \node (ForceVector) {};
    % 	\node [smallblock, right = 1.5 em of ForceVector,  minimum width = 2.5em, minimum height = 1.75 em] (Attitude) {\small f2q};
    	%\node [below of=Attitude,node distance=1.5cm] (Ref_traj) {};
    	
    	\node (Attitude) {};
    	\node [sum, right = 2. em of Attitude] (sum1) {};
    	\node[draw = white] at (sum1.center) {$+$};
    	\node [smallblock, right = 2 em of sum1,  minimum width = 2.5em, minimum height = 2 em] (Cont1) {\small$G_q$};
    	\node [sum, right = 2 em of Cont1] (sum2) {};
    	\node[draw = white] at (sum2.center) {$+$};
    	\node [smallblock, right = 2 em of sum2,  minimum width = 2.5em, minimum height = 1.75 em] (Cont2) {\small$G_\omega$};
    	\node [right = 1.25 em of Cont2] (output) {};
    	
    % 	\draw[->] (ForceVector.center) node [xshift=0.6em, yshift = 0.8em] {\scriptsize$\vect f_\rmc\resolvedin{E}{\rm sp}$} -- (Attitude.west);
    	\draw[->] (Attitude.east) -- (sum1.west) node [xshift=-1.75em, yshift = 1.2em] {\scriptsize \twolineentry{Attitude}{setpoint}};
    	\draw [->] ([yshift = -1.2em]sum1.south) -- node [xshift=2.8 em, yshift = -1.4 em] {\scriptsize \twolineentry{Attitude}{measurement}} node [xshift=1.4em, yshift = 0.2em]{$-$} (sum1.south);
    	\draw [->] (sum1.east) -- (Cont1.west);
    	\draw [->] (Cont1.east) -- (sum2.west);
    	\draw [->] ([yshift = -1.2em]sum2.south) -- node [xshift=2.9 em, yshift = -1.4 em] {\scriptsize \twolineentry{Angular rate}{measurement}} node [xshift=1.4em, yshift = 0.2em]{$-$} (sum2.south);
    	\draw [->] (sum2.east) -- (Cont2.west);
    % 	\draw [->] ([yshift = 0.75em]Cont1.north) -- node [xshift=-1.5em, yshift = 0.9em] {\scriptsize$\dot \Psi_{\rm sp,ff}$} (Cont1.north);
    	\draw [->] (Cont1.east) -- ([xshift=1em]Cont1.east) |- node [xshift=1.8 em, yshift = -0.25em] {\scriptsize\twolineentry{Angular rate}{setpoint}} ([yshift = 0.75em]Cont2.north) -- (Cont2.north);
    	\draw [->] (Cont2.east) -- (output.center) node [xshift=0.25 em, yshift = 1.5 em] {\scriptsize\threelineentry{Moment}{vector}{setpoint}};
    \end{tikzpicture}
    }
    \vspace{-1em}\caption{\footnotesize Attitude controller architecture.}
    \label{fig:PX4_autopilot_inner_loop}
\end{figure}
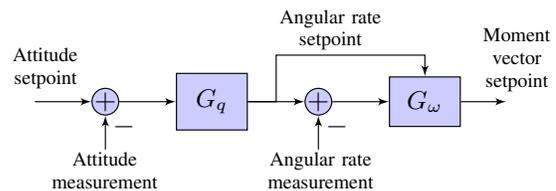

The control system implemented in the PX4 autopilot thus consists of 27 gains.
In particular, the outer loop includes three gains in $G_r$ and nine gains in $G_v$; and
the inner loop includes three gains in $G_q$ and 12 gains in $G_\omega.$
In standard practice, these 27 gains are manually tuned and require considerable expertise. 

% Finally, note that $G_r, G_v, G_q, $ and $G_\omega$ are MIMO controllers with three inputs and three outputs. 
% Denoting the $i$th component of the input and the output of $G_r$ at step $k$ by $z_{r,i,k}$ and $u_{r,i,k},$ it follows that 
% \begin{align}
%     u_{r,i,k}
%         =
%             G_r(\shiftq)
%             z_{r,i,k}.
% \end{align}
% Similarly, 

% Let the input and the output of the controller $G_{\rms}$ be denoted by $z_{\rms,k}$ and $u_{\rms,k},$ where $\rms $
% Note that $z_\rms \in \BBR^3$ and $u_\rms.$

\section{Adaptive Autotuner}
\label{sec:autotuner}
% The autotuner consists of a data-driven learning algorithm and a learning trajectory.
The adaptive autotuner is constructed by replacing the fixed-gain controllers in the autopilot described in the previous section with adaptive controllers that are updated by retrospective cost optimization. 
In particular, each fixed-gain controller described in the previous section is replaced by an adaptive controller parameterized with the same structure.
% that is, a P controller is replaced by an adaptive P controller, and a PID controller is replaced by an adaptive PID controller. 
% 
Specifically, denoting the $i$th component of the input and the output of $G_r$ at step $k$ by $z_{r,i,k}$ and $u_{r,i,k},$ the velocity setpoint 
\begin{align}
    u_{r,i,k}
        =
            G_r(k,\shiftq)
            g(z_{r,i,k}),
\end{align}
where $G_r(k,\shiftq) = \text{diag}(\theta_{r,1,k}, \theta_{r,2,k}, \theta_{r,3,k}),$ 
$g(z)$ is a error-normalization function, 
and the gains $\theta_{r,i,k}$ are updated by the RCAC algorithm described in the next section. 
The functions $g(z)$ used in this work are given in Table \ref{tab:RCAC_Hyperparameters}.
Similarly, 
\begin{align}
    u_{v,i,k}
        =
            G_v(k,\shiftq)
            g(z_{v,i,k}),
\end{align}
where, for $i\in \{1,2,3 \}$, the entry $(i,i)$ 
\begin{align}
    G_{v,i,i}(k,\shiftq) 
        =
            \frac{\theta_{v,i+1}}{\shiftq} +
            \frac{\theta_{v,i+2}}{\shiftq-1} +
             \frac{\theta_{v,i+3}(\shiftq-1)}{\shiftq^2} .
\end{align}
The controllers $G_q(k,\shiftq)$ and $G_{\omega}(k,\shiftq)$ are similarly parameterized by the gains $\theta_q \in \BBR^3$ and $\theta_\omega \in \BBR^{12}.$

The adaptive gains $\theta_r, \theta_v, \theta_q,$ and $\theta_\omega$ are updated in a \textit{learning} trajectory, which consists of flying through the waypoints described in Table \ref{tab:learning_trajectory},
parameterized by $z_{\rm hov}, x_{\rm inc}, y_{\rm inc},$ and $z_{\rm inc}$.
% 
% 
%  the autotuning phase, the multicopter is commanded to follow a learning trajectory designed to excite 
%all control modes and, in turn, 
% the inputs of all adaptive controllers.
%
% Specifically, the multicopter is in through the waypoints described in Table \ref{tab:learning_trajectory}.
%
Note that the waypoints are denoted by coordinates $(x,y,z,\psi),$ which correspond to the three components of the position vector and azimuth of the multicopter in the East($x$)-North($y$)-Up($z$) (ENU) coordinate frame.
%
% Furthermore, the learning trajectory is parameterized by the hover altitude $z_{\rm hov},$ and the increments in the three coordinates $x_{\rm inc}, y_{\rm inc},$ and $z_{\rm inc}.$
% 

\begin{table}[h]
    \caption{Waypoints in the learning trajectory}
    \label{tab:learning_trajectory}
    \centering
    \renewcommand{\arraystretch}{1.2}
    \begin{tabular}{|c|l|l|}
        \hline
        \multicolumn{1}{|c|}{\textbf{Waypoint}} &  \multicolumn{1}{|c|}{\textbf{Coordinate}}  &  \multicolumn{1}{|c|}{\textbf{Remark}} \\ \hhline{|=|=|=|}
        1 & $(0,0,z_{\rm hov},0)$        &   Take-off  \\ \hline
        2 & $(0,0,z_{\rm hov} + z_{\rm inc},0)$        &   Move along $+z$ direction \\ \hline
        3 & $(0,0,z_{\rm hov},0)$        &   Move along $-z$ direction\\ \hline
        4 & $(0,y_{\rm inc},z_{\rm hov},0)$        &   Move along $+y$ direction \\ \hline
        5 & $(0,-y_{\rm inc},z_{\rm hov},0)$       &   Move along $-y$ direction\\ \hline
        6 & $(0,0,z_{\rm hov},0)$        &   Move along $+y$ direction\\ \hline
        7 & $(x_{\rm inc},0,z_{\rm hov},0)$        &   Move along $+x$ direction\\ \hline
        8 & $(-x_{\rm inc},0,z_{\rm hov},0)$       &   Move along $-x$ direction\\ \hline
        9 & $(0,0,z_{\rm hov},0)$        &   Move along $+x$ direction\\ \hline
        10 & $(0,0,z_{\rm hov},\pi/2)$     &   Turn counterclockwise\\ \hline
        11 & $(0,0,z_{\rm hov},\pi)$     &   Turn counterclockwise\\ \hline
        12 & $(0,0,z_{\rm hov}, 0)$      &   Turn clockwise\\ \hline
        13 & $(0,0,0,0)$       &   Land \\ \hline
    \end{tabular}
\end{table}

% At the end of the learning trajectory, 
The gains $\theta_r, \theta_v, \theta_q,$ and $\theta_\omega$ obtained at the end of the learning trajectory are the \textit{autotuned gains} and the autopilot implemented with the autotuned gains is the \textit{autotuned autopilot}.

% The gains $\theta_r, \theta_v, \theta_q,$ and $\theta_\omega,$ updated by RCAC in real-time, at the end of the trajectory are the autotuned gains and the autopilot . 

% optimized by the adaptive control algorithm, shown in Section \ref{sec:PID_Algo}, at the end of the trajectory are the autotuned gains. 
%
% Throughout the learning trajectory, the controller gains shown in Section \ref{sec:PID_Algo} will be optimized by the adaptive control algorithm.
%
% Once the end of the learning trajectory is reached, the optimized controller gains are saved and will be used as the autotuned gains.

\section{RCAC Algorithm}
\label{sec:PID_Algo}
This section briefly reviews the retrospective cost adaptive control (RCAC) algorithm. 
% In this work, retrospective cost adaptive control (RCAC) technique that is used to update the control law in a sampled-data feedback loop.
RCAC is described in detail in \cite{rahmanCSM2017} and its extension to digital PID control is given in \cite{rezaPID}.
% The algorithm in \cite{rezaPID} is specialized for a PID control structure as shown in \cite{goel2020adaptive}.
% 
% Consider the control law
% \begin{align}
%     u_k 
%         =
%             \phi_k \theta_k,
%     \label{eq:uk_reg}
% \end{align}
% where, for all $k\ge0$, 
% the regressor $\phi_k \in \BBR^{l_u \times l_\theta}$ contains the measurements and $l_\theta$ depends on the structure of the controller.
% % 
% The controller coefficients $\theta_k \in \BBR^{l_\theta}$ are optimized by RCAC as described below. 

Consider a SISO PID controller with a feedforward term
\begin{align}
    u_k
        &=
            K_{\rmp,k} g(z_{k-1}) +
            K_{\rmi,k} \gamma_{k-1} 
            \nn \\ &\quad +
            K_{\rmd,k} (g(z_{k-1}) - g(z_{k-2})) +
            K_{{\rm ff},k} r_k
            ,
    \label{eq:uk_PID}
\end{align}
where $K_{\rmp,k}, K_{\rmi,k}, K_{\rmd,k}, $ and $K_{{\rm ff},k}$ are time-varying gains to be optimized, 
$z_k$ is an error variable,
$r_k$ is the feedforward signal, 
and, for all $k\ge0$,
%%\vspace{-5pt}
\begin{align}
    \gamma_k 
        \isdef
            \sum_{i=0}^{k} g(z_{i}).
\end{align}
Note that the integrator state is computed recursively using $\gamma_k = \gamma_{k-1} + g(z_{k-1})$.
For all $k\ge0$, 
the control law can be written as 
\begin{align}
    u_k 
        =
            \phi_k \theta_k,
    \label{eq:uk_reg}
\end{align}
where 
the regressor $\phi_k$ 
and 
the controller gains $\theta_k$ 
% in \eqref{eq:uk_PID} 
are 
\begin{align}
    \phi_k
        \isdef
            \matl{c}
                g(z_{k-1}) \\
                \gamma_{k-1} \\
                g(z_{k-1}) - g(z_{k-2}) \\
                r_k
            \matr^\rmT, \quad
    \theta_k
        \isdef
            \matl{c}
                K_{\rmp,k} \\
                K_{\rmi,k} \\
                K_{\rmd,k} \\
                K_{{\rm ff},k}
            \matr .
            % \in \BBR^{4}.
    \label{eq:phi_theta_def}
\end{align}
Note that the P, PI, or PID controllers can be parameterized by appropriately defining $\phi_k$ and $\theta_k.$ 
Various MIMO controller parameterizations are shown in \cite{goel_2020_sparse_para}.

To determine the controller gains $\theta_k$, let $\theta \in \BBR^{l_\theta}$, and consider the \textit{retrospective performance variable} defined by
%%\vspace{-5pt}
\begin{align}
    \hat{z}_{k}(\theta)
        \isdef
            g(z_k) + 
            \sigma (\phi_{k-1} \theta - u_{k-1}),
    \label{eq:zhat_def}
\end{align}
where $\sigma \in \BBR.$
The sign of $\sigma$ is the sign of the leading numerator coefficient of the transfer function from $u_k$ to $z_k.$
Furthermore, define the \textit{retrospective cost function} $J_k \colon \BBR^{l_\theta} \to [0,\infty)$ by
%%\vspace{-10pt}
\begin{align}
    J_k(\theta) 
        &\isdef
            \sum_{i=0}^k
                \hat{z}_{i}(\theta) ^\rmT 
                R_z
                \hat{z}_{i}(\theta)
                 +
                (\phi_k \theta)^\rmT
                R_u
                (\phi_k \theta)
                \nn \\ &\quad \quad +
                (\theta-\theta_0)^\rmT 
                P_0^{-1}
                (\theta-\theta_0),
    \label{eq:RetCost_def}
\end{align}
where $\theta_0\in\BBR^{l_\theta}$ is the initial vector of PID gains and $P_0\in\BBR^{l_\theta\times l_\theta}$ is positive definite.
%
% For all examples in this paper, we set $\theta_0 = 0$; however, $\theta_0$ can be initialized to nonzero gains in practice if desired.

\begin{proposition}
    Consider \eqref{eq:uk_reg}--\eqref{eq:RetCost_def}, 
    where $\theta_0 \in \BBR^{l_\theta}$ and $P_0 \in \BBR^{l_\theta \times l_\theta}$ is positive definite. 
    Furthermore, for all $k\ge0$, denote the minimizer of $J_k$ given by \eqref{eq:RetCost_def} by
    \begin{align}
        \theta_{k+1}
            \isdef
                \underset{ \theta \in \BBR^n  }{\operatorname{argmin}} \
                J_k({\theta}).
        \label{eq:theta_minimizer_def}
    \end{align}
    Then, for all $k\ge0$, $\theta_{k+1}$ is given by 
    \begin{align}
        \theta_{k+1} 
            &=
                \theta_k  - 
                 \sigma P_{k+1}\phi_{k-1}^\rmT R_z
                 [ z_k + \sigma(\phi_{k-1} \theta_k - u_{k-1}) ]
                 \\ &\quad \quad 
                 - 
                 P_{k+1}\phi_{k}^\rmT
                 R_u \phi_{k} \theta_k 
                 , \label{eq:theta_update}
    \end{align}
    where 
    \begin{align}
        P_{k+1} 
            &=
                P_{k}
                -  
                P_k  
            \Phi_k ^\rmT 
            \left( 
                \bar R ^{-1} +  
                \Phi_k
                P_k
                \Phi_k ^\rmT 
            \right)^{-1}
            \Phi_k,
        \label{eq:P_update_noInverse}
    \end{align}
    and
    \begin{align}
        \Phi_k
            \isdef 
            \matl{c}
                \sigma \phi_{k-1} \\
                \phi_k
            \matr,
        \quad 
        \bar R
            \isdef 
                \matl{cc}
                    R_z & 0 \\
                    0   & R_u
                \matr.
    \end{align}
    % \begin{align}
    %     P_{k+1} 
    %         &=
    %             P_{k}
    %             -  \frac
    %                 { P_{k}\phi_{k-1}^\rmT  \phi_{k-1} P_{k} }
    %                 { 1 +   \phi_{k-1} P_{k} \phi_{k-1}^\rmT  }.
    %     \label{eq:P_update_noInverse}
    % \end{align}
\end{proposition}
\begin{proof}
See \cite{AseemRLS}
\end{proof}

Finally, the control is given by
\begin{align}
    u_{k+1} = \phi_{k+1} \theta_{k+1}.
\end{align}

\section{Experimental Results}
\label{sec:flight_tests}
This section describes the experimental results obtained in the simulation environment and in the physical flight tests. 
In both the simulation and the physical flight tests, the autotuner is used to tune the 27 controller gains described in Section \ref{sec:PX4_autopilot} using the learning trajectory described in Section \ref{sec:autotuner}.
Note that in the autotuning mode, all of the gains in both loops are initialized at zero. 
%
%The \textit{autotuned autopilot} is then compared with the \textit{default autopilot} by flying a test trajectory, whose waypoints are generated by a second-order Hilbert Curve.
%
% Let the \textit{autotuned autopilot} and the \textit{default autopilot} be the autopilot described in Section \ref{sec:PX4_autopilot} using the tunes gains and the stock PX4 gains recpectively.
%
% After the autotuner is finished, the performance of the \textit{autotuned autopilot} and the \textit{default autopilot} is compared by flying a test trajectory, whose waypoints are generated by a second-order Hilbert Curve.
% Specifically, the autotuner is first tested with a quadcopter simulation in jMAVSim. 
% 
% This section describes the experimental setup used to demonstrate the autotuner. % 
% In the autotuning mode, the quadcopter is commanded to follow the learning trajectory, whose waypoints are given in Table \ref{tab:learning_trajectory}.
% 
% 
% In the simulation environment, the hyperparameters of the adaptive controller are tuned.
% 
% The performance of the autotuned autopilot is then compared with the performance of the default PX4 autopilot by flying a trajectory whose waypoints are generated using a second-order Hilbert curve. 
% 
%%%%%First, the autotuner is tested in a simulation environment, where the quadcopter is simulated using jMAVSim. 
% 
% In this work, the autotuner is implemented in the PX4 flight stack and the quadcopter is simulated using jMAVSim. 
% 
The hyperparameters $R_u,P_0,$ and the error-normalization function used in the RCAC algorithm are shown in Table \ref{tab:RCAC_Hyperparameters}. 
Furthermore, $\sigma= R_z=1$ in all four controllers. 
% 
% Figure xxx shows the learning-trajectory following response and the gains adapted by the autotuner. 

\begin{table}[h]
     \caption{\footnotesize RCAC hyperparameters in the adaptive autopilot.}
    \label{tab:RCAC_Hyperparameters}
    \centering
    \renewcommand{\arraystretch}{1.2}
    \begin{tabular}{|c|l|l|l|}
        \hline
        \multicolumn{1}{|c|}{\textbf{Controller}}  & \multicolumn{1}{|c|}{${P_0}$}  & \multicolumn{1}{|c|}{${R_u}$} & \multicolumn{1}{|c|}{$g(z)$}
        \\ \hhline{|=|=|=|=|}
        $G_r$ & $0.01$ & $0.01$ & $z$
        \\ \hline
        $G_v$ & $0.1$ & $0.01$ & ${\rm erf} \left( \tfrac{\sqrt{\pi}}{2} z\right)$
        \\ \hline
        $G_q$ & $1$ & $0.001$ & $z$
        \\ \hline
        $G_\omega$  & $0.0001$ & $0.1$ & ${\rm erf} \left( \tfrac{\sqrt{\pi}}{2} z\right)$
        \\ \hline
    \end{tabular}
    % \vspace{-1em}
\end{table}

\subsection{Simulation Flight Tests}
First, the autotuner is tested in a simulation environment, where the quadcopter is simulated using jMAVSim.
In the learning trajectory,  $z_{\rm hov} = 5$ m, and $x_{\rm inc} = y_{\rm inc} = z_{\rm inc} = 5$ m. 
The guidance law generates the setpoints using $v_{\rm max} = 6 $ m/s, $a_{\rm max} = 1$ m/${\rm s}^2$, $\dot{\psi}_{\rm max} = 2 $ rad/s, $\ddot{\psi}_{\rm max} = $ 0.5 rad/${\rm s}^2.$
%
% The parameters for the learning trajectory are $z_{\rm hov} = 5$ m, and $x_{\rm inc} = y_{\rm inc} = z_{\rm inc} = 5$ m. 
%

% Figures \ref{fig:trainingTraj_setpoints_bigger} and \ref{fig:trainingTraj_topdown} show the response of the quadcopter in the learning trajectory.
Figure \ref{fig:trainingTraj_setpoints_bigger} shows the response of the quadcopter in the learning trajectory.
% 
% learning trajectory and the response of the quadcopter obtained with the jMAVSim simulator.
Note that the response improves as RCAC re-optimizes the controllers continuously in real time.
Figure \ref{fig:trainingTraj_gains} shows the gains adapted by the RCAC algorithm for all four controllers.
The gains at the end of the learning trajectory are saved, are the autotuned gains and constitute the autotuned autopilot.

\begin{figure}[h]
    \centering
    \includegraphics[width=0.7\columnwidth]{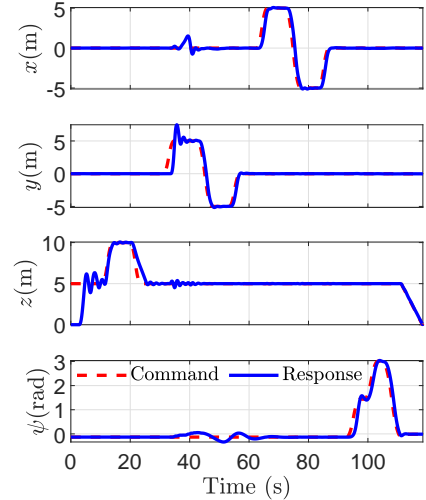}
    \caption{
    \footnotesize
    Simulation results. 
    Response of the quadcopter in the learning trajectory in the jMAVSim simulator.
    % The subplots show the position and yaw response of the quadcopter. 
    }	
    \label{fig:trainingTraj_setpoints_bigger}
\end{figure}
% \begin{figure}[h!]
%     \centering
%     \includegraphics[width=\columnwidth]{Figures/trainingTrajPlots/SITL/trainingTraj_topdown.eps}
%     \caption{Simulation results. Response of the quadcopter in the learning trajectory in the jMAVSim simulator. }	
%     \label{fig:trainingTraj_topdown}
% \end{figure}
\begin{figure}[h!]
    \centering
    \includegraphics[width=0.7\columnwidth]{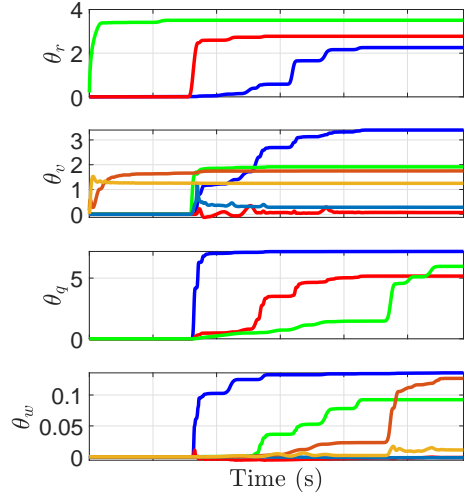}
    \caption{
    Simulation results. Autoilot gains adapted by the RCAC algorithm during the learning trajectory in the jMAVSim simulator. 
    % Autotuning response of the quadcopter obtained in the jMAVSim simulator.
    }	
    \label{fig:trainingTraj_gains}
\end{figure}

Next, the quadcopter is commanded to follow a test trajectory, whose waypoints are generated using a second-order Hilbert curve, first with the \textit{default autopilot}, and next with the \textit{autotuned autopilot}. 
The \textit{default autopilot} gains and the actuator constraints in PX4 are specified in the
\verb|mc_pos_control_params.c|,
\verb|mc_att_control_params.c|, and
\verb|mc_rate_cont-| \verb|rol_params.c| files
%\footnote{\href{https://github.com/ankgoel8188/Firmware}{https://github.com/ankgoel8188/Firmware}}.
\footnote{\href{https://github.com/JAParedes/PX4-Autopilot/tree/RCAC\_MC\_AutoTuner}{https://github.com/JAParedes/PX4-Autopilot/tree/RCAC\_MC\_AutoTuner}}.
Note that the results in this paper are based on PX4 version V1.11.3.
% 
% 
% The waypoints for the test trajectory are generated using a second-order Hilbert curve. 
% Note that the gains adapt 
% Note that the gains adapt when the error in the corresponding controller increases. 
Figure \ref{fig:sitl_hilbertCurve_response} shows the response of the quadcopter obtained with the {autotuned autopilot} and the {default autopilot}. 
Note that the autotuned autopilot outperforms the default autopilot in terms of position tracking error. 

\begin{figure}[h!]
    \centering
    \includegraphics[width=\columnwidth]{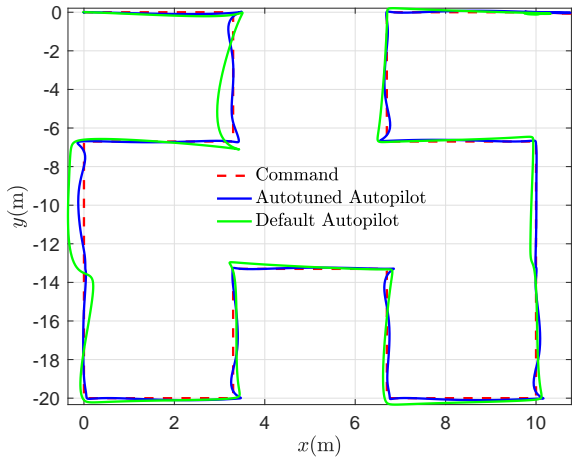}
    \caption{
    Simulation results. 
    Response of the quadcopter in the test trajectory.
    Note that the autotuned autopilot uses the gains tuned during the learning trajectory.
    }
    \label{fig:sitl_hilbertCurve_response}
\end{figure}

In order to quantify and compare the performance of the \textit{autotuned autopilot} with the \textit{default autopilot}, a position-tracking cost variable $J$ defined by
\begin{align}
    J \isdef \frac{1}{T} \sum_{i=0}^N z_i^\rmT z_i,
    \label{eq:cost_compare}
\end{align}
where $z_i\in \BBR^3$ is the position error and  $T$ is the total flight time in seconds, is computed for each test.
%
% Let $J_\rmD$ and $J_\rmA$ be the cost corresponding to equation \eqref{eq:cost_compare} associated with the default and autotuned autopilots respectively.
%
% Figure \ref{fig:hilbertCurve_histogram} shows the cost $J$ computed for the test trajectory flown in simulation tests.
% each controller in the autotuned and the default autopilot. 
%
Table \ref{tab:comp_table} shows the cost default autopilot cost $J_\rmD$ and the autotuned autopilot cost $J_\rmA.$
% computed for the test trajectory flown in simulation tests.
% 
Note that the {autotuned autopilot} is relatively 38 $\%$ better than the {default autopilot} in simulation.
% \begin{figure}[h!]
%     \centering
%     \includegraphics[width=\columnwidth]{Figures/hilbertPlots/SITL/hilbertCurve_histogram.eps}
%     \caption{Cost variable given by \eqref{eq:cost_compare} for the autotuned and the default autopilot obtained in the test trajectory.  
%     % \textbf{Joon, replace ylabel with $J$ and xlabel with $G_r, ...$. }
%     }
%     \label{fig:hilbertCurve_histogram}
% \end{figure}

\begin{table}[h!]
    \caption{Cost variable given by \eqref{eq:cost_compare} for the autotuned and the default autopilots obtained in the test trajectory.  
    }
    \label{tab:comp_table}
    \centering
    \renewcommand{\arraystretch}{1.25}
    \setlength{\tabcolsep}{2pt}
    \resizebox{\columnwidth}{!}{%
    \begin{tabular}{|c|c|c|c|}
    \hline
        \multicolumn{1}{|c|}{\textbf{Test Type}} & \multicolumn{1}{|c|}{$\begin{array}{c}\textbf{Default} \\ \textbf{Autopilot}\\ \textbf{Cost} \ \bm{(J_\rmD)} \end{array}$} & \multicolumn{1}{|c|}{$\begin{array}{c}\textbf{Autotuned} \\ \textbf{Autopilot}\\ \textbf{Cost} \ \bm{(J_\rmA)} \end{array}$} & \multicolumn{1}{|c|}{$\begin{array}{c}\textbf{Relative} \\ \textbf{Improvement}\\ \bm{(J_\rmD - J_\rmA) / J_\rmD}  \end{array}$} \\
    \hhline{|=|=|=|=|}
        $\begin{array}{c} {\rm Simulation} \\ {\rm (jMAVSim)} \end{array}$ & 1.222 & 0.752 & 38.4 \% \\
    \hline 
        $\begin{array}{c} {\rm Flight \ Test} \\ {\rm (M}\text{-}{\rm Air)} \end{array}$ & 0.321 & 0.218 & 32.3 \% \\
    \hline
    \end{tabular}
    }
\end{table}

Next, the sensitivity of the autotuning process to changes in the physical parameters of the quadcopter is investigated.
To do so, the mass of the quadcopter is scaled in the dynamic model simulated by jMAVSim. 
Specifically, the mass of the quadcopter is multiplied by $\alpha \in \{0.5, 0.75, 1, 1.25, 1.5\}.$
For each value of $\alpha,$ the autopilot gains are tuned using the autotuner.
Note that the RCAC hyperparameters are not changed. 
% 
% Note that in each case, the initial autopilot gains are set to zero, and the RCAC hyperparameters are not changed. 
%
The performance of the autotuned and the default autpilots is compared for each value of $\alpha$ by flying the test trajectory defined by the second-order Hilbert curve.
% 
%Figure xxx shows the trajectory-following response of the quadcopter for various values of $\alpha.$
% 
Figure \ref{fig:hilbertCurve_histogram} shows the cost $J$ computed for the position controller in each flight flown with the autotuned and the default autopilot. 
Note that the autotuned autopilot is re-tuned for each value of $\alpha,$ whereas default autopilot is fixed for all values of $\alpha.$

% The results are shown in Figure \ref{fig:hilbertCurve_histogram}, which shows the cost $J$ computed for the autotuned and the default autopilots after flying the test trajectory  for each value of $\alpha.$

% The default controller gains in PX4 are well-tuned for the jMAVSim simulator. 
% To investigate the potential improvements in the performance, the mass of the quadcopter in the jMAVSim simulation is multiplied by a scalar $\alpha>0$ to degrade the performance of the default PX4 autopilot.  
% % 
% Note that the baseline performance is obtained by setting $\alpha=1$ and the degraded autopilot performance by setting $\alpha\neq 1.$
% % 
% Figure xxx shows the trajectory-following response of the jMAVSim model for various values of $\alpha.$
\begin{figure}[h!]
    \centering
    \includegraphics[width=0.9\columnwidth]{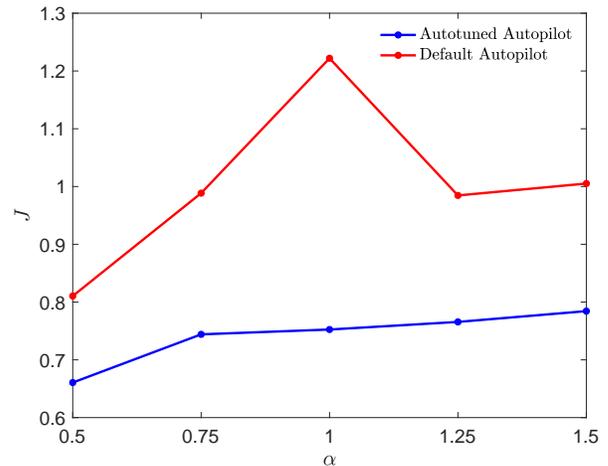}
    \caption{
    Simulation results. 
    Position tracking cost in the test trajectory for the autotuned and the default autopilots obtained.
    Note that the mass of the quadcopter is multiplied by the scalar $\alpha$ in the jMAVSim simulation. 
    % The scalar $\alpha$ multiplies the mass of the quadcopter model is scaled by $\alpha.$  
    }
    \label{fig:hilbertCurve_histogram}
\end{figure}

% \section{Flight Tests}
% \label{sec:flight_tests}
% This section describes the results of the flight tests conducted in the M-Air facility at The University of Michigan, Ann Arbor with the Holybro X500 quadcopter. 
% 

\subsection{Physical Flight Tests}
Next, the autotuner is tested in physical flight tests conducted in the M-Air facility at the University of Michigan, Ann Arbor with the Holybro X500 quadcopter frame.
The M-Air is equipped with a motion capture system that allows high-precision position and attitude measurements.
% 
% using motion capture cameras for precise state extimation.
%
% Figure \ref{fig:quadcopter} shows the quadcopter, the M-Air facility, and the experimental setup used for the physical flight tests.
%
% \begin{figure}[h]
%     \centering
%     \includegraphics[width = \columnwidth]{Figures/Quadcopter_Mair_MOCAP.png}
%     \caption{
%     \footnotesize
%     Physical flight tests experimental setup. 
%     a) shows the quadcopter used for testing,
%     b) shows the M-Air testing facility,
%     c) shows wireless bridge used to transmit motion capture measurements to the quadcopter, and 
%     d) shows the ground station for all wireless communications.
%     % c) shows the device use to receive motion capture (MOCAP) position and attitude data for precise state estimation via WiFi, and
%     % d) shows the computer systems used for transmission of MOCAP data and monitoring flight tests via telemetry.
%     }
%     \label{fig:quadcopter}
% \end{figure}

In the learning trajectory, $z_{\rm hov} = 2$ m, $x_{\rm inc} = y_{\rm inc} = 4$ m, and $z_{\rm inc} = 1$ m.
The guidance law generates the setpoints using 
$v_{\rm max} = 3 $ m/s, 
$a_{\rm max} = 1.2$ m/${\rm s}^2$, 
$\dot{\psi}_{\rm max} = 2 $ rad/s, 
$\ddot{\psi}_{\rm max} = $ 0.5 rad/${\rm s}^2.$
%
% Figures \ref{fig:trainingTraj_setpoints_bigger_mair} and \ref{fig:trainingTraj_topdown_mair} show the response of the quadcopter in the learning trajectory.
Figure \ref{fig:trainingTraj_setpoints_bigger_mair} shows the response of the quadcopter in the learning trajectory.
Figure \ref{fig:trainingTraj_gains_mair} shows the gains adapted by the RCAC algorithm for all four controllers.

% For the physical flight tests, the guidance law parameters are $v_{\rm max} = 3 $ m/s, $a_{\rm max} = 1.2$ m/${\rm m}^2$, $|\dot{\psi}_{\rm max}| = 2 $ deg/s, $|\ddot{\psi}_{\rm max}| = $ 0.5 deg/${\rm m}^2.$
%

% 

%

Next, the quadcopter is commanded to follow a test trajectory with the {autotuned autopilot} and the {default autopilot} in the M-Air facility. 
The waypoints for the test trajectory are generated using a second-order Hilbert curve. 
% Note that the gains adapt 
% Note that the gains adapt when the error in the corresponding controller increases. 
Figure \ref{fig:hilbertCurve_response_mair} shows the response of the quadcopter obtained with the {autotuned autopilot} and the {default autopilot} during the physical flight tests.
Table \ref{tab:comp_table} shows the position-tracking cost $J$ computed for the test trajectory flown during the physical flight tests.
Note that the {autotuned autopilot} is relatively 32 $\%$ better than the {default autopilot}.

\begin{figure}[h!]
    \centering
    \includegraphics[width=0.7\columnwidth]{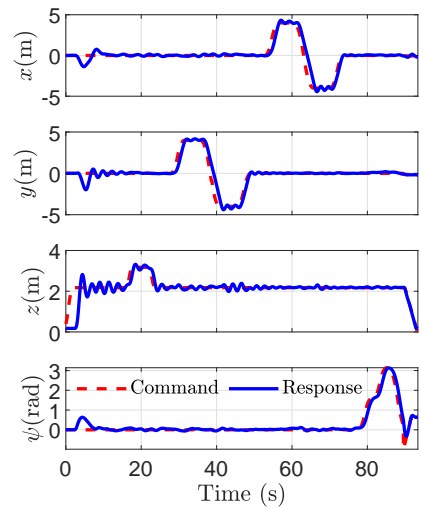}
    \caption{
    Flight test results. Response of the quadcopter in the learning trajectory at the M-Air facility.
    % The subplots show the position and yaw response of the quadcopter.  
    }
    \label{fig:trainingTraj_setpoints_bigger_mair}
\end{figure}
% \begin{figure}[h!]
%     \centering
%     \includegraphics[width=\columnwidth]{Figures/trainingTrajPlots/FlightTest/trainingTraj_topdown.eps}
%     \caption{Flight test results. Response of the quadcopter in the learning trajectory at the M-Air facility. }	
%     \label{fig:trainingTraj_topdown_mair}
% \end{figure}
\begin{figure}[h!]
    \centering
    \includegraphics[width=0.7\columnwidth]{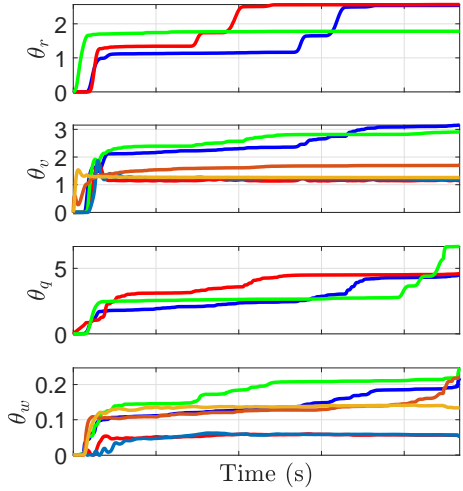}
    \caption{
    Flight test results. Autopilot gains adapted by the RCAC algorithm during the learning trajectory at the M-Air facility.
    }
    \label{fig:trainingTraj_gains_mair}
\end{figure}

\begin{figure}[h!]
    \centering
    \includegraphics[width=\columnwidth]{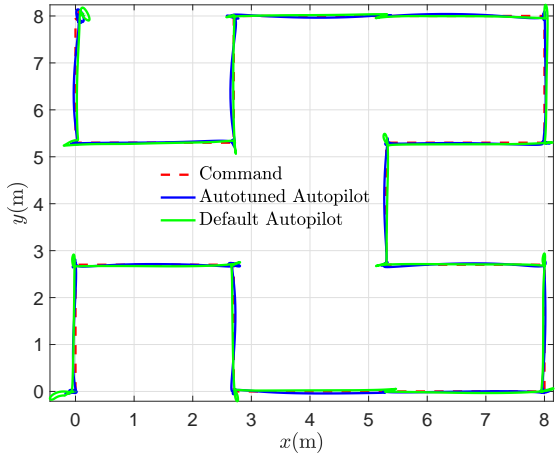}
    \caption{ Flight test results. 
    Response of the quadcopter in the test trajectory.
    Note that the autotuned autopilot uses the gains tuned during the learning trajectory.
    }
    \label{fig:hilbertCurve_response_mair}
\end{figure}

% \begin{figure}[h!]
%     \centering
%     \includegraphics[width=\columnwidth]{Figures/hilbertPlots/FlightTest/hilbertCurve_histogram.eps}
%     \caption{Cost variable given by \eqref{eq:cost_compare} for the autotuned and the default autopilot obtained in the test trajectory.  
%     \textbf{Joon, replace ylabel with $J$ and xlabel with $G_r, ...$. }
%     }
%     \label{fig:hilbertCurve_histogram}
% \end{figure}

% \clearpage
\section{Conclusions and Future Work}
\label{sec:conclusions}

This paper presented an adaptive autotuner that can automatically tune a multicopter autopilot without using any modeling information about the multicopter.  
The adaptive autotuner, implemented in the PX4 flight stack,  consists of adaptive PID controllers that are updated by the retrospective cost adaptive control algorithm by flying a single learning trajectory.

The adaptive autotuner was validated in a simulation environment, and its performance was compared against the default autopilot by scaling the mass of the quadcopter. 
The adaptive autotuner outperformed the default autotuner in all cases without any change in the hyperparameters of the adaptive algorithm. 
Next, using the same hyperparameters used for the simulation, the adaptive autotuner was used to tune the autopilot for the X500 Holybro quadcopter in the M-Air facility, where the adaptive autotuner outperformed the default autotuner in the test trajectory.

Future work will focus on using the adaptive autopilot to improve the performance of the default autopilot flying a multicopter with an unknown suspended payload and under various actuator failure conditions.

\begin{appendices} 
\section{}\label{app:guid_law}

% \section{Appendix}

% Let $x_\rmr$ denote the distance to the next waypoint.
Let $r_\rmc \in \BBR^3$ denote the current position, $r_\rmd \in \BBR^3$ denote the desired position, and $d$ denote the distance between the current and the desired position, that is,  $d \isdef ||r_\rmd - r_\rmc ||_2.$
Let $a_{\rm max}>0$ and $v_{\rm max}>0.$
The desired trajectory consists of an acceleration phase, a cruise phase, and a deceleration phase along the straight line between $r_\rmc$ and $r_\rmd.$ 

Let the setpoints $r(t) \in \BBR^3$ be given by the guidance law
\begin{align}
    r(t) 
        =
            r_\rmc + (r_\rmd-r_\rmc) 
            \frac{q(t)}{d},
    \label{eq:guidance_law}
\end{align}
where $q(t)$ is calculated as shown below. 
Using the fact that $r(0)=r_\rmc$ and $r(T_\rmf)=r_\rmd,$ where $T_\rmf$ is the total flight time, it follows that 
$q(0) = 0$ and $q(T_\rmf)=d.$
Note that $q(t)$ is the distance between $r(t)$ and $r_\rmc.$
% 
% Note that \eqref{eq:guidance_law} is the guidance law. 
% 

% Let $q(t)$ denote the distance along $r_\rmd-r_\rmc.$ 
% Thus, $q(0) = 0$ and $q(T_\rmf)=d,$ where $T_\rmf$ is the total flight time, and is computed as shown below.

Let $T_1$ denote the time to reach the speed $v_{\rm max}$ at maximum acceleration $a_{\rm max}.$
Thus, $T_1= \dfrac{v_{\rm max}}{a_{\rm max}}$ and 
$     q(T_1)
        =
            \dfrac{1}{2} a_{\rm max} T_1^2
        =
            \dfrac{1}{2} \dfrac{v_{\rm max}^2}{a_{\rm max}}.
$
% Assuming maximum acceleration along the direction $r_\rmd-r_\rmc,$ the quadcopter reaches the speed $v_{\rm max}$ in $T_1$ seconds, where 
% \begin{align}
%     T_1= \frac{v_{\rm max}}{a_{\rm max}},
% \end{align}
% and
% \begin{align}
%     q(T_1)
%         =
%             \frac{1}{2} a_{\rm max} T_1^2
%         =
%             \frac{1}{2} \frac{v_{\rm max}^2}{a_{\rm max}}.
% \end{align}

% Suppose that QC accelerate at the maximum acceleration, cruises at the maximum velocity, and then decelerates at maximum acceleration to reach the desired waypoint. 
% Let $T_1$ be the time to reach the maximum velocity. 
% Then, 
% \begin{align}
%     T_1 =   \frac{v_{\rm max}}{a_{\rm max}}.
% \end{align}
% and 
% \begin{align}
%     x(T_1)
%         =
%             \frac{1}{2} a_{\rm max} T_1^2
%         =
%             \frac{1}{2} \frac{v_{\rm max}^2}{a_{\rm max}}
% \end{align}

First, consider the case where $d \geq 2 q(T_1).$ 
In this case, the quadcopter cruises at $v(t)=v_{\rm max}$ for $t\in [T_1, T_2],$ where $t=T_2$ denotes the time at which deceleration phase starts, and is calculated as shown below.  
Note that 
\begin{align}
    q(T_2) 
    %     &=
    %         q(T_1) + v_{\rm max} (T_2 - T_1), 
    % \\
        &=    
            \frac{1}{2} \frac{v_{\rm max}^2}{a_{\rm max}} + v_{\rm max} (T_2 - T_1).
\end{align}
% 
% 
% Assuming that the quadcopter decelerates at the maximum rate $a_{\rm max}$ to reach the waypoint $r_\rmd,$ 
Since the quadcopter takes $T_1$ seconds to decelerate from $v_{\rm max}$ to $0$ speed at constant deceleration $a_{\rm max},$ 
the total flight time is $T_\rmf = T_1+T_2$ and 
\begin{align}
    q(T_3) = \frac{v_{\rm max}^2}{a_{\rm max}} + v_{\rm max} (T_2 - T_1) = d.
    \label{eq:r(T_3)}
\end{align}
It follows from \eqref{eq:r(T_3)} that 
$T_2 
        =
            \dfrac{d}{v_{\rm max}} $
% \begin{align}
%     T_2 
%         =
%             \frac{d}{v_{\rm max}} .
%             % T_1 + 
%             % \left( \frac{d}{v_{\rm max}} - \frac{v_{\rm max}}{a_{\rm max}} \right).
% \end{align}
and thus, the total flight time 
$T_\rmf
        = 
            \dfrac{v_{\rm max}}{a_{\rm max}} + 
            \dfrac{d}{v_{\rm max}}.$
% \begin{align}
%     T_\rmf
%         = 
%             \frac{v_{\rm max}}{a_{\rm max}} + 
%             \frac{d}{v_{\rm max}}.
% \end{align}
The distance $q(t)$ is thus given by
\begin{align}
    q(t)
        =
    \begin{cases}
        \frac{1}{2} a_{\rm max} t^2, & t\in [0, T_1), \\
        \frac{1}{2} a_{\rm max} T_1^2 + 
        v_{\rm max} (t - T_1), & t\in [T_1, T_2 ), \\
        \frac{1}{2} a_{\rm max} T_1^2 + 
        v_{\rm max} (T_2 - T_1) 
        \\ \quad  
        + 
        v_{\rm max} (t- T_2) - 
        \frac{1}{2} a_{\rm max } (t-T_2)^2
        , & t\in [T_2, T_\rmf ).
    \end{cases}
\end{align}

Next, consider the case where $d < 2 q(T_1).$
In this case, the quadcopter velocity does not reach $v_{\rm max}$, that is, there is no cruise phase. 
% 
% In this case, assume that the QC accelerates to the maximum velocity and decelerates to zero velocity in the same time. 
Let $\bar T_1$ denote the time instant at which maximum velocity $\bar v_{\rm max}$ is reached. 
Thus, 
$\bar T_1 =   \dfrac{\bar v_{\rm max}}{a_{\rm max}}$
% \begin{align}
%     \bar T_1 =   \frac{\bar v_{\rm max}}{a_{\rm max}},
% \end{align}
and
$q(\bar T_1)
        =
            \dfrac{1}{2} \dfrac{\bar v_{\rm max}^2}{a_{\rm max}}.$
% \begin{align}
%     q(\bar T_1)
%         =
%             \frac{1}{2} \frac{\bar v_{\rm max}^2}{a_{\rm max}}.
% \end{align}
Assuming that the quadcopter decelerates at the same rate to rest, it follows that $q(\bar T_1)=d/2,$ and thus 
$\bar v_{\rm max}
        =
            \sqrt{d a_{\rm max}},$
% \begin{align}
%     \bar v_{\rm max}
%         &=
%             \sqrt{d a_{\rm max}},
% \end{align}
$\bar T_1
        =
            \sqrt{ \dfrac{d}{ a_{\rm max}}}, 
$and $
    T_\rmf = 2 \sqrt{ \dfrac{d}{ a_{\rm max}}}.$
% \begin{align}
%     \bar T_1
%         &=
%             \sqrt{ \frac{d}{ a_{\rm max}}}, 
%         \quad 
%     T_\rmf = 2 \sqrt{ \frac{d}{ a_{\rm max}}}.
% \end{align}
The distance $q(t)$ is thus given by
% For $t\in (0, \bar T_1),$
\begin{align}
    q(t)
        =
    \begin{cases}
            \frac{1}{2} a_{\rm max} t^2, & t\in [0, \bar T_1), 
        % \\
            % \frac{1}{2} a_{\rm max} \bar T_1^2 + 
            % \frac{1}{2}
            % \left(
            %     v_{\rm max} \bar T_1- (v_{\rm max} -a_{\rm max}(t-\bar T_1)) (2 \bar T_1 -t) 
            % \right)
            % % \frac{1}{2} a_{\rm max} (t-\bar T_1)^2
            % , & t\in [\bar T_1, 2 \bar T_1 ), 
        \\
            \frac{1}{2} a_{\rm max} \bar T_1^2 + 
            \bar v_{\rm max} (t-\bar T_1) 
            \\ \quad 
            - 
            \frac{1}{2} a_{\rm max } (t-\bar T_1)^2
            % \frac{1}{2} a_{\rm max} (t-\bar T_1)^2
            , & t\in [\bar T_1, 2 T_\rmf ).
    \end{cases}
\end{align}

\end{appendices}

% \bibliographystyle{IEEEtran}
% \bibliography{PX4bib}
%\renewcommand*{\bibfont}{\small}
%\printbibliography

\bibliographystyle{IEEEtran}
\bibliography{main_quad_tuner_arxiv}

\end{document}